\documentclass{article}
\usepackage{fullpage,amsfonts,amssymb,amsmath,amsthm,hyperref,latexsym,graphicx,color,cancel,xstring,mleftright}
\usepackage{authblk}

\usepackage[T1]{fontenc}
\usepackage{graphicx}

\usepackage[utf8]{inputenc} 
\usepackage[T1]{fontenc}    
\usepackage{hyperref}       
\usepackage{url}            
\usepackage{booktabs}       
\usepackage{amsfonts}       
\usepackage{nicefrac}       
\usepackage{microtype}      
\usepackage{xcolor}         

\usepackage{amsmath}
\usepackage{cryptocode}

\usepackage[capitalize]{cleveref}
\DeclarePairedDelimiter\abs{\lvert}{\rvert}

\newcommand{\ot}{\otimes}

\renewcommand{\epsilon}{\varepsilon}

\usepackage{comment}

\usepackage{mathtools}
\usepackage{enumerate}

\usepackage{amsfonts}

\usepackage{amssymb}
\usepackage{amsmath}

\usepackage{algorithm}
\usepackage{algpseudocode}

\usepackage{dirtytalk}

\usepackage{subfig}
\usepackage{booktabs}
\usepackage{siunitx}

\usepackage{datetime}

\newdateformat{monthdayyeardate}{%
  \monthname[\THEMONTH]~\THEDAY, \THEYEAR}

\usepackage{bbm}
\usepackage{ dsfont }
\usepackage{wrapfig}

\usepackage{tikz}

\usepackage[full]{complexity}

\newcommand{\negl}{\mathsf{negl}}

\usetikzlibrary{quantikz}

\newcommand{\PRF}{\textsf{PRF}}

  \theoremstyle{definition}
\newtheorem{theorem}{Theorem}[section]

\newtheorem{definition}[theorem]{Definition}
\newtheorem{fact}[theorem]{Fact}

\newtheorem{lemma}[theorem]{Lemma}

\newcommand{\adv}{\mathbf{A}}

\newcommand{\alice}{\mathbf{A}}
\newcommand{\bob}{\mathbf{B}}

\newlength\myindent
\usepackage[disable]{todonotes}
\newcommand{\ct}{\textsf{ct}}
\newcommand{\pt}{\textsf{pt}}
\newcommand{\pk}{\textsf{pk}}
\newcommand{\sk}{\textsf{sk}}
\newcommand{\Hilb}{\mathcal{H}}
 \newcommand{\CC}{\mathbb{C}}
 \newcommand{\Gen}{\textsf{Gen}}
\newcommand{\Enc}{\textsf{Enc}}
\newcommand{\Dec}{\textsf{Dec}}

\newcounter{mycomment}
\newcommand{\comm}[2]{
\refstepcounter{mycomment}
{%
    \todo[author = \textbf{#1~\#~\themycomment}, color={red!100!green!35}, fancyline, size = \footnotesize]{%
        #2}%
    }
}
\newcommand{\MW}[1]{\comm{MW}{#1}}
\renewcommand{\MW}[1]{}
\begin{document}
%
\title{A Simple Construction of Quantum Public-Key Encryption from Quantum-Secure One-Way Functions}
\date{}
%

%
\author[1]{Khashayar Barooti}
\author[2]{Giulio Malavolta}
\author[3]{Michael Walter}
\affil[1]{EPFL, Lausanne, Switzerland}
\affil[2]{Max-Planck Institute in Security and Privacy, Bochum, Germany}
\affil[3]{Ruhr-Universit\"{a}t Bochum, Bochum, Germany}

\maketitle
\begin{abstract}
Quantum public-key encryption [Gottesman; Kawachi et al., Eurocrypt'05] generalizes public-key encryption (PKE) by allowing the public keys to be quantum states.
Prior work indicated that quantum PKE can be constructed from assumptions that are potentially weaker than those needed to realize its classical counterpart.
In this work, we show that quantum PKE can be constructed from any quantum-secure one-way function.
In contrast, classical PKE is believed to require more structured assumptions.
Our construction is simple, uses only classical ciphertexts, and satisfies the strong notion of \emph{CCA security}.
\end{abstract}

\section{Introduction}
Quantum public-key encryption was proposed by Gottesman~\cite{gottesman} and Kawachi et al.~\cite{DBLP:journals/iacr/KawachiKNY06} as a generalization of the standard notion of public-key encryption, allowing public keys to be quantum states.
More specifically, this primitive allows Alice to locally generate (many copies of) a state~$\ket{\pk}$ and upload it to some certificate authority.
Later on, Bob can query the certificate authority to retrieve a copy of~$\ket{\pk}$ and use it to send a private message to Alice.
\MW{Add a sentence saying that one can consider both the setting where Bob's ciphertext is classical or quantum.}
Similarly to the classical setting, quantum PKE assumes that the certificate authority provides Bob with the correct information (in this case the state~$\ket{\pk}$), but does not otherwise make any assumption on the behavior of the certificate authority, who could try to learn the secret key of Alice in some arbitrary way.
However, contrary to the classical case, since quantum states cannot in general be copied, one has to assume that Alice uploads many copies of~$\ket{\pk}$, if she wants to establish a secure channel with multiple parties.
\MW{Or if Bob wants to send many messages (although here one can use the trick in \cref{subsec:context}, so good to forward reference that discussion from here).}
In spite of this limitation, quantum PKE is still an interesting object to study:
(i)~Because of the use of quantum information, quantum PKE may be realizable from weaker computational assumptions than standard (classical) PKE, or perhaps even unconditionally.
(ii)~In contrast to quantum key-distribution (QKD) protocols~\cite{DBLP:journals/sigact/BennettB87}, which require more interaction, quantum PKE preserves the interaction pattern of classical PKE, and thus enables \emph{round-optimal} secure communication.
Yet, the current state of affairs of quantum PKE leaves open many questions regarding the minimal assumptions needed to construct this primitive.
Existing proposals~\cite{DBLP:journals/iacr/KawachiKNY06} rely on ad-hoc assumptions that are seemingly insufficient for classical PKE, but do not give a clear complexity-theoretic characterization of this primitive.
There are even proposals of unconditionally secure quantum PKE~\cite{gottesman}, although without security proofs.
We note that conjecturing unconditional security for quantum PKE is at the very least plausible -- after all, QKD does achieve information-theoretic security (assuming authenticated channels).

\subsection{Our results}

In this work, we show that quantum-secure one-way functions are sufficient to build quantum public-key encryption schemes.
While elementary, our construction satisfies the strong notion of \emph{CCA security}, and the ciphertexts are classical.
Our results should be contrasted with the case of classical PKE, where one-way functions are widely believed to be insufficient for realizing this primitive, and in fact black-box separations are known~\cite{DBLP:conf/stoc/ImpagliazzoR89}.
Thus, our result also implies a black-box separation between classical and quantum PKE.

\begin{theorem}[Informal]\label{thm:informal positive}
If quantum-secure one-way functions exist, then there exists a (CCA-secure) quantum PKE scheme with classical ciphertexts.
\end{theorem}

\subsection{Open problems}
Our results demonstrate that in fact quantum PKE can be realized only assuming the existence of one-way functions.
However, in contrast to classical cryptography, in quantum cryptography one-way functions are not considered to be a minimal assumption.
The work of Kretschmer~\cite{DBLP:conf/tqc/Kretschmer21} shows that there exists an oracle relative to which one-way functions do not exist, but pseudorandom states~\cite{DBLP:conf/crypto/JiL018} do.
Thus a question left open by our work is whether quantum PKE can be constructed from presumably weaker assumptions than quantum-secure one-way functions.

\subsection{Quantum PKE in context}\label{subsec:context}
Although the notion of quantum PKE is not an original contribution of our work, we feel compelled to discuss its relation to related primitives in classical and quantum cryptography.
While quantum PKE mimics the interaction pattern of traditional (classical) PKE, the presence of quantum information in the public keys introduces some important conceptual differences.
As alluded at earlier, an important difference is that public keys can no longer be copied.
This means that Alice (the receiver) must upload many copies of the quantum state~$\ket{\pk}$ to the certificate authority, and each copy gets ``consumed'' once Bob uses it to encrypt a message for Alice.
One possible way to mitigate this limitation is to let Alice and Bob use their first message to exchange a secret key, and then continue the remainder of the interaction using standard symmetric encryption.
In this way, only one state is used up in this interaction.

Another important point of having a quantum state for a public key is that it becomes less obvious how to check whether the certificate authority sent us the ``correct'' public key.
This is not a problem unique to the quantum setting, since also in the classical case the certificate authority must be trusted to supply the public keys correctly.
One standard approach to address this problem is to have multiple authorities storing the same keys, so that one can check their honesty by just comparing the public keys that we receive.
With quantum information, the same idea can be implemented by using the SWAP test, which allows comparing two unknown quantum states~\cite{buhrman2001quantum}.
This problem and its solution based on the SWAP test were already observed by Gottesman~\cite{gottesman}.

Quantum PKE can also be compared with quantum key distribution (QKD).
On the one hand, quantum PKE has a single round of interaction (from Alice to Bob and back), thus satisfying a \emph{stronger} notion of efficiency.
On the other hand, quantum PKE requires Alice to keep a long-term secret key that she would use for all subsequent communications, whereas in QKD there is no requirement to keep a state across different executions.
More significantly, in quantum PKE the certificate authority is guaranteed to correctly deliver the public key~$\ket{\pk}$ from Alice to Bob, whereas in QKD the eavesdropper can behave arbitrarily during all rounds of the protocol.
Thus, quantum PKE satisfies a \emph{weaker} security notion, although the difference is slightly more nuanced than what appears superficially -- while the certificate authority is required to deliver the correct state to Bob, it can do arbitrary computations locally to try to recover the secret key.
In fact, the security definition provides it with many copies of~$\ket{\pk}$ that it could potentially use to learn the secret key. A more thorough discussion on this aspect, along with the modelling choices for the attacker, is given in~\cref{sec:quantum_PKE}.

Given the above discussion it is natural to ask whether one can generalize the notion of quantum PKE to allow for quantum secret keys. If one allows the public key and the secret key to form a (possibly entangled) quantum state, then quantum PKE can be realized unconditionally via quantum teleportation~\cite{TP}.

\section{Preliminaries}\label{sec:prelims}
In this section, we provide some preliminary background on quantum mechanics and quantum information.
The state space of a quantum system can be characterized by a Hilbert space~$\Hilb$.
For a more in-depth introduction to quantum information, we refer the reader to \cite{DBLP:books/daglib/0046438}.
The state of a machine can be represented as a \emph{density matrix}, a positive semi-definite operator of trace one, on $\Hilb$.
We call a state \emph{pure} if this operator is a rank one projector, i.e., equal to $\proj{\psi}$, where $\ket{\psi}\in \Hilb$ is a unit vector. This allows us to represent pure states as unit vectors of $\Hilb$ instead of density operators.

When $\Hilb = (\CC^2)^{\ot n}$, the quantum system consists of~$n$ \emph{quantum bits} or \emph{qubits}.
The standard product basis~$\ket x = \ket{x_1} \ot \dots \ot \ket{x_n}$ of~$\Hilb$ is labeled by bit strings~$x \in \{0,1\}^n$ and is known as the \emph{computational basis}.
For sake of convenience, we often leave out $\ot$ and we also write $\ket{0}$ instead of $\ket{0}\otimes \dots \otimes \ket{0}$.
Every pure state can be represented as $\ket\psi = \sum_{x\in\{0,1\}^n} \alpha_x \ket x$, where the~$\alpha_x$ are called amplitudes and satisfy $\sum_x \abs{\alpha_x}^2 = 1$.
When measuring the qubits of a quantum system in this state, the probability of the measurement outcome being~$x$ is given by $\abs{\alpha_x}^2$.
When the amplitudes are all equal, i.e., the state is at a uniform superposition, we drop the normalization~$2^{-n/2}$ and simply write~$\sum_{x \in \{0,1\}^n} \ket x$.
Next, we state a well-known fact about the quantum evaluation of classical circuits.


\begin{fact}\label{thm:quantum}
Let $f\colon\{0,1\}^n \to \{0,1\}^m$ be a function which is efficiently computable by a classical circuit.
Then there exists a unitary~$U_f$ on $(\CC^2)^{\ot n+m}$ which is efficiently computable by a quantum circuit (possibly using ancillas) such that, for all~$x \in\{0,1\}^n$ and~$y\in\{0,1\}^m$,
\begin{align*}
    U_f\colon \ket{x}\ket{y} \mapsto \ket{x}\ket{y \oplus f(x)}.
\end{align*}
\end{fact}

Next, we recall the one-way to hiding lemma~\cite{O2H}.
\begin{lemma}[One-way to hiding]\label{lemma:o2h}
Let $G,H: X\to Y$ be random functions and $S \subset X$ an arbitrary set with the condition that $\forall x\notin S, G(x) = H(x)$, and let $z$ be a random bitstring. Further, let $\alice^H(z)$ be a quantum oracle algorithm that queries $H$ with depth at most $d$. Define $\bob^H(z)$ to be an algorithm that picks $i\in[d]$ uniformly, runs $\alice^H(z)$ until just before its $i^{th}$ round of queries to $H$ and measures all query input registers in the computational basis and collects them in a set $T$. Let
\begin{align*}
    P_{\text{left}} = \Pr[ 1\gets \alice^H(z)], \quad
    P_{\text{right}} = \Pr[1\gets \alice^G(z)], \quad
    P_{\text{guess}} = \Pr[S\cap T \neq \emptyset | T\gets \bob^H(z)].
\end{align*}
Then we have that
\begin{align}
    |P_{\text{left}} - P_{\text{right}}| \leq 2d \sqrt{P_{\text{guess}}} \quad\text{and}\quad |\sqrt{P_{\text{left}}} - \sqrt{P_{\text{right}}}| \leq 2d \sqrt{P_{\text{guess}}}
\end{align}
\end{lemma}

\comm{GM}{@Khashayar: Can you check if we need all of the stuff above? We should remove the notions that we don't use here.}

\subsection{Quantum-secure pseudorandom functions}\label{subsec:prf}
Our construction relies on a \emph{pseudorandom function} (PRF)~\cite{DBLP:journals/jacm/GoldreichGM86}.
This is a keyed function, denoted~$\PRF$, that can be evaluated in polynomial time satisfying a certain security property.
In this work we require $\PRF$ to be \emph{quantum-secure}, which, loosely speaking, says that an adversary with oracle access to $\PRF$ cannot distinguish it from a truly random function, even given superposition queries.
It is known that quantum-secure PRFs can be constructed from any quantum-secure one-way function~\cite{DBLP:journals/jacm/Zhandry21}.

\begin{definition}[Quantum-secure PRF]
We say that a keyed function~$\PRF$ is a \emph{quantum-secure pseudorandom function (PRF)} if, for any quantum polynomial time (QPT) adversary~$\adv$, we have
\begin{align*}
  \abs*{
    \Pr\mleft[ 1 \gets \adv(1^\lambda)^{\PRF_k} \mright]
  - \Pr\mleft[ 1 \gets \adv(1^\lambda)^{f} \mright]
  }
  \leq \mu(\lambda),
\end{align*}
where $k\xleftarrow{\$}\{0,1\}^\lambda$, $f$ is a truly random function, and the oracles can be accessed in superposition, that is, they implement the following unitaries
\[
  \ket{x}\ket{z} \xmapsto{U_{\PRF_k}} \ket{x}\ket{z\oplus \PRF_k(x)}
\quad\text{and}\quad
  \ket{x}\ket{z} \xmapsto{U_f} \ket{x}\ket{z\oplus f(x)},
\]
respectively.
\end{definition}

\subsection{Quantum public-key encryption}\label{sec:quantum_PKE}

We start by formalizing the notion of a quantum public-key encryption (PKE) scheme~\cite{DBLP:journals/iacr/KawachiKNY06}.
For convenience, we consider a PKE with binary message space~$\{0,1\}$, however the scheme can be generically upgraded to encrypt messages of arbitrary length, via the standard hybrid encryption paradigm. This transformation is known to preserve CPA security, by a standard hybrid argument. Classically, it is known that bit-encryption is also complete for CCA security~\cite{MS09}, however we leave the proof of such a statement in the quantum settings as ground for future work.


\begin{definition}[Quantum PKE]\label{def:pke}
A \emph{quantum public key encryption (PKE)} scheme is defined as a tuple~$\Gamma = (\Gen,\Enc,\Dec)$ such that:
\begin{itemize}
    \item $(\ket\pk,\sk) \gets \Gen(1^{\lambda})$ is a QPT algorithm which outputs a \emph{pure} quantum state $\ket\pk$ and a bit string $\sk$;
    \item $\ct \gets \Enc(\ket\pk, \pt)$ is a QPT algorithm which, given a bit $\pt\in\{0,1\}$, outputs a quantum state $\ct$ (that needs not be pure);
    \item $\pt \gets \Dec(\sk, \ct)$ is a QPT algorithm which outputs a bit $\pt\in\{0,1\}$.
\end{itemize}
If the ciphertext is classical then we call~$\Gamma$ a \emph{quantum PKE scheme with classical ciphertexts}.
In general, we say that~$\Gamma$ has \emph{correctness error~$\epsilon$} (which can be a function of~$\lambda$) if for all $\pt\in\{0,1\}$ we have
\begin{align*}
  \Pr\bigl[ \Dec(\sk , \Enc(\ket\pk,\pt)) = \pt \;:\; (\ket\pk,\sk) \gets \Gen(1^{\lambda})\bigr] \geq 1-\epsilon.
\end{align*}
Finally, we say~$\Gamma$ is \emph{correct} if it has negligible correctness error.
\end{definition}


Next we define the notion of CCA security~\cite{dolev1991non}. The version that is going to be relevant for us is the definition of CCA security for quantum PKE with classical ciphertext, which we present below. We also explicitly mention here that CCA security is not easy to define when ciphertexts are quantum states, and it is currently an open question to find the correct analogue of CCA security in the quantum settings~\cite{boneh2013secure,gagliardoni2016semantic,carstens2021relationships}. However, in this work, we will only consider CCA security for schemes with \emph{classical} ciphertexts, and therefore the decryption oracle is queried only classically. 


\begin{definition}[CCA Security]\label{def:cca}
We say $\Gamma = (\Gen,\Enc,\Dec)$ is \emph{CCA-secure}, if for any polynomial~$n=n(\lambda)$, and any QPT adversary $\adv = (\adv_0, \adv_1)$ we have
\begin{align*}
    \Pr\mleft[b\gets \adv_1^{\Dec^*(\sk, \cdot)}\mleft(\ct , \ket{\mathsf{st}}\mright):
    \begin{array}{l}
        (\ket\pk,\sk) \xleftarrow{} \Gen(1^\lambda)  \\
        \ket{\mathsf{st}} \xleftarrow{} \adv_0^{{\Dec(\sk, \cdot)}}(\ket\pk^{\otimes n})\\
         b \xleftarrow{\$} \{0,1\} \\
         \ct \xleftarrow{} \Enc(\ket\pk , b)
    \end{array}
    \mright] \leq 1/2 + \epsilon(\lambda)
\end{align*}
where $\epsilon$ is a negligible function, and the oracle $\Dec^*$ is defined as $\Dec$, except that it returns $\bot$ on input the challenge ciphertext $\ct$.
\end{definition}
A few remarks about the above definitions are in order.
First of all, we would like to stress that \cref{def:pke} \emph{crucially} imposes that the public key $\pk$ must be a pure state. If the public key was allowed to be the classical mixture, there there is a trivial scheme that satisfies this notion. Namely, the public key consists of a pair
\[
\mathsf{SK}.\Enc(0; r) \text{ and } \mathsf{SK}.\Enc(1; r)
\]
where $r$ is uniformly sampled and $\mathsf{SK}.\Enc$ is a secret-key encryption scheme. Formally, the public key is modelled as a classical mixture over the randomness $r$. The encryption algorithm would the just select one out of these two ciphertexts, depending on the input message. Note that this scheme is fully classical. This is the approach suggested by~\cite{MY22} and, as the authors also point out, it can hardly be considered a public key encryption scheme, since it does not even protect against a \emph{passive} certificate authority, who can break the scheme by simply looking at the public key.
While we do no formalize the notion of security against passive (possibly quantum) adversaries here, we believe that our definition, by forcing the state to be pure, models the intuition behind a semi-honest certificate authority who is trusted to deliver the correct state, but can otherwise do arbitrary computations on the public key.

An alternative definition that also captures this intuition is the notion of \emph{specious adversaries}~\cite{2QPC}, where the adversary is actually allowed to maul the public keys arbitrarily, conditioned on the fact that the state that it returns must be indistinguishable from the original one. Arguably, this is the ``right'' formalization of semi-honest adversaries in the quantum settings and in fact it is not satisfied by the trivial construction outlined above. This definition is slightly more general than~\cref{def:cca}, which on the other hand has the advantage to be simpler to state.


\section{CCA-secure quantum PKE from one-way functions}\label{sec:quantum_PKE_CCA}
In the following section, we describe our quantum PKE scheme and show that it satisfies the strong notion of CCA security.
The construction relies on a quantum-secure pseudorandom function
\[
  \PRF\colon \{0,1\}^\lambda \times \{0,1\}^\lambda \to \{0,1\}^{3\lambda}
\]
which, as mentioned earlier in \cref{subsec:prf}, can be constructed from any quantum-secure one-way function.
Then our quantum PKE scheme~$\Gamma=(\Gen,\Enc,\Dec)$ is defined as follows:
\begin{itemize}
\item The key generation algorithm $\Gen(1^\lambda)$ samples two keys $k_0 \xleftarrow{\$} \{0,1\}^\lambda$ and $k_1 \xleftarrow{\$} \{0,1\}^\lambda$, then it prepares the states
\begin{align*}
  \ket{\pk_0} = \sum_{x\in \{0,1\}^{\lambda}} \ket{x,\PRF_{k_0}(x)}
\quad\text{and}\quad
  \ket{\pk_1} = \sum_{x\in \{0,1\}^{\lambda}} \ket{x,\PRF_{k_1}(x)}.
\end{align*}
Note that both states are efficiently computable since the $\PRF$ can be efficiently evaluated in superposition in view of \cref{thm:quantum}.
The quantum public key is then given by the pure state $\ket\pk = \ket{\pk_0} \ot \ket{\pk_1}$, whereas the classical secret key consists of the pair $\sk = (k_0,k_1)$.

\item Given a message $\pt \in \{0,1\}$, the encryption algorithm $\Enc(\ket\pk, \pt)$ simply measures $\ket{\pk_{\pt}}$ in the computational basis, and outputs the measurement outcome as the \emph{classical} ciphertext $\ct=(x,y)$.

\item Given the ciphertext $\ct=(x,y)$, the decryption algorithm $\Dec(\sk, \ct)$ first checks whether $\PRF_{k_0}(x) = y$ and returns $0$ if this is the case.
Next, it checks whether $\PRF_{k_1}(x) = y$ and returns $1$ in this case.
Finally, if neither is the case, the decryption algorithm returns $\bot$.
\end{itemize}

Next, we establish correctness of this scheme.

\begin{theorem}\label{thm:correctness}
If $\PRF$ is a quantum-secure one-way function, then the quantum PKE scheme $\Gamma$ is correct.
\end{theorem}
\begin{proof}
Observe that the scheme is perfectly correct if the ranges of $\PRF_{k_0}$ and $\PRF_{k_1}$ are disjoint.
By a standard argument, we can instead analyze the case of two truly random functions~$f_0$ and~$f_1$, and the same will hold for~$\PRF_{k_0}$ and~$\PRF_{k_1}$, except on a negligible fraction of the inputs.
Fix the range of~$f_0$, which is of size at most~$2^\lambda$.
Then the probability that any given element of~$f_1$ falls into the same set is at most $2^{-2\lambda}$, and the desired statement follows by a union bound.
\end{proof}

Finally, we show that the scheme is CCA-secure.

\begin{theorem}\label{thm:CCA}
If $\PRF$ is a quantum-secure one-way function, then the quantum PKE scheme $\Gamma$ is CCA-secure.
\end{theorem}
\begin{proof}
It suffices to show that the CCA experiment with the bit~$b$ fixed to~$0$ is indistinguishable from the same experiment but with~$b$ fixed to~$1$.
To this end we consider a series of hybrids, starting with the former and ending with the latter:
\begin{itemize}
    \item \textbf{Hybrid 0:} This is the original CCA experiment except that the bit $b$ fixed to $0$.
    \item \textbf{Hybrid 1:} In this (inefficient) hybrid, we modify hybrid~0 to instead compute $\ket{\pk_0}$ as
    \[
    \ket{\pk_0} = \sum_{x\in \{0,1\}^{\lambda}} \ket{x,f(x)},
    \]
    where $f$ is a truly uniformly random function.
\end{itemize}
The indistinguishability between these two hybrids follows by a standard reduction against the quantum security of~$\PRF$:
To simulate the desired $n$ copies of $\ket{\pk_0}$, and to answer decryption queries (except the one that contains the challenge ciphertext), the reduction simply queries the oracle provided by the $\PRF$ security experiment (possibly in superposition). Note that whenever the oracle implements $\PRF$, then the view of the distinguisher is identical to hybrid $0$, whereas if the oracle implements a truly random function, then the view of the distinguisher is identical to hybrid $1$.

\begin{itemize}
    \item \textbf{Hybrid 2:} In this (inefficient) hybrid, we modify hybrid~1 such that the challenge ciphertext is sampled as
    \[
    x \xleftarrow{\$} \{0,1\}^{\lambda} \quad\text{and}\quad y \xleftarrow{\$} \{0,1\}^{3\lambda}.
    \]
\end{itemize}
The indistinguishability of hybrids~1 and~2 follows from the one-way to hiding lemma (\cref{lemma:o2h}). Let $H$ be such that $H(x) = y$ and for all $x'\neq x$ we set $H(x') = f(x')$, and let $S = \{x\}$. Let $\alice$ be the adversary playing the security experiment. We claim that $\alice^f$ is the adversary playing in hybrid~$1$ whereas $\alice^H$ corresponds to the adversary playing hybrid~$2$: Observe that the public keys can be simulated with oracle access to $f$ ($H$, respectively) by simply querying on a uniform superposition of the input domain, whereas the decryption queries can be simulated by query basis states. Importantly, for all queries after the challenge phase, the adversary is not allowed to query $x$ to $\Dec^*$. Hence the set $T$, collected by $\bob$ is a set of at most $n$ uniform elements from the domain of $f$, along with $Q$ basis states, where $Q$ denotes the number of queries made by the adversary to the decryption oracle \emph{before} the challenge ciphertext is issued. By a union bound
\[P_{\text{guess}} =\Pr[T \cap \{x\} \neq \emptyset] \leq \frac{(n+Q)}{2^\lambda}=\negl(\lambda)\]
since $x$ is uniformly sampled. Applying~\cref{lemma:o2h}, we deduce that $|P_{\text{left}}-P_{\text{right}}|$ is also negligible, i.e., which bounds the distance between the two hybrids.

\begin{itemize}
    \item \textbf{Hybrid 3:} In this (efficient) hybrid, we modify hybrid~2 to compute $\ket{\pk_0}$ by using the pseudorandom function~$\PRF_{k_0}$ instead of the truly random function~$f$.
    That is, we revert the change done in hybrid~$1$.
\end{itemize}
Indistinguishability follows from the same argument as above.

\begin{itemize}
    \item \textbf{Hybrid 4:} In this (inefficient) hybrid, we modify hybrid~3 to compute $\ket{\pk_1}$ as
    \[
    \ket{\pk_1} = \sum_{x\in \{0,1\}^{\lambda}} \ket{x,f(x)}
    \]
    where $f$ is a truly uniformly random function.
\end{itemize}
Indistinguishability follows from the same argument as above.

\begin{itemize}
    \item \textbf{Hybrid 5:} In this (inefficient) hybrid, we modify hybrid~4 by fixing the bit $b$ to $1$ and computing the challenge ciphertext honestly, i.e., as
        \[
    x \xleftarrow{\$} \{0,1\}^{\lambda} \quad\text{and}\quad y =f(x).
    \]
\end{itemize}
Indistinguishability follows from the same argument as above.

\begin{itemize}
    \item \textbf{Hybrid 6:} In this (efficient) hybrid, we modify hybrid~5 to compute $\ket{\pk_1}$ by using the pseudorandom function~$\PRF_{k_1}$ instead of the truly random function~$f$.
    That is, we revert the change done in hybrid~$4$.
\end{itemize}
Indistinguishability follows from the same argument as above. The proof is concluded by observing that the last hybrid is identical to the CCA experiment with the bit~$b$ fixed to $1$.
\end{proof}

\bibliographystyle{splncs04}
\bibliography{ref}

\end{document}